\documentclass[letterpaper, abstracton, DIV11, 12pt, titlepage=false]{scrartcl}

\usepackage[ruled]{algorithm2e}

\SetAlFnt{\small}
\SetAlCapFnt{\small}
\SetAlCapNameFnt{\small}
\SetAlCapHSkip{0pt}
\IncMargin{-\parindent}
\let\chapter\undefined

\usepackage[usenames,dvipsnames]{xcolor}
\usepackage{footnote}
\usepackage{enumerate}
\usepackage[round]{natbib}
\usepackage[english]{babel}
\usepackage{amsfonts}
\usepackage{amsmath}
\usepackage{amssymb}
\usepackage{booktabs}
\usepackage{cancel}
\usepackage{amsthm}
\usepackage{multirow}
\usepackage{mathabx}
\usepackage{dsfont}
\usepackage{graphicx}
\usepackage{tikz}
\usepackage{pifont}
\usepackage{bm}
\usepackage{wrapfig}
\usepackage{microtype}
\usepackage{setspace}
\usepackage[colorlinks=true,bookmarks=true, pdftex, citecolor = blue, linkcolor = blue,urlcolor = blue]{hyperref}

\onehalfspace

\DisableLigatures[f]{encoding = *, family = * }

\def\bf{\normalfont\bfseries}

\changenotsign

\definecolor{darkgreen}{RGB}{40,150,70}

\theoremstyle{plain}% default
\newtheorem{theorem}{Theorem}
\newtheorem{lemma}{Lemma}

\newtheorem{corollary}{Corollary}

\theoremstyle{definition}
\newtheorem{definition}{Definition}

\newtheorem{axiom}{Axiom}
\newcommand{\ourrep}{}

\theoremstyle{remark}
\newtheorem{remark}{Remark}

\begin{document}

% Partial Strategyproofness: Relaxing Strategyproofness of Assignment Mechanisms
% Partial Strategyproofness: A New Relaxed Notion of Strategyproofness for Assignment Mechanisms
% Partial Strategyproofness: A Parametric Relaxation of Strategyproofness for Assignment Mechanisms
% an axiomatic approach to characterizating and relaxing strategyproofness in one-sided matching
% Partial Strategyproofness: Relaxing Strategyproofness of Assignment Mechanisms
{\setstretch{1}
\title{%
\LARGE{%
An Axiomatic Decomposition of Strategyproofness 
for Ordinal Mechanism \\
with Indifferences}\thanks{%
\scriptsize{%
Department of Informatics, University of Zurich, Switzerland,
email: \{mennle, seuken\}@ifi.uzh.ch.
For updates see www.ifi.uzh.ch/ce/publications/PSP.pdf. 
We would like to thank Baharak Rastegari for insightful discussions.
Part of this research was supported by the Hasler Foundation under grant \#12078 and the SNSF (Swiss National Science Foundation) under grant \#156836.}}}
\author{%
Timo Mennle \\ University of Zurich
\and Sven Seuken \\ University of Zurich }
\date{First version: July 14, 2014 \\
This version: \today}
%\date{February 13, 2014}
\maketitle

%\vspace{-0.3in}
%\begin{center}
%\begin{Large}
%Working Paper!
%\end{Large}
%\end{center}

\begin{abstract}
We study mechanism which operate on ordinal preference information (i.e., rank ordered lists of alternatives) on the full domain of weak preferences that admits indifferences. 
We present a novel decomposition of strategyproofness into three axioms: 
\emph{separation monotonic}, 
\emph{separation upper invariant}, 
and \emph{separation lower invariant}. 
Each axiom is a natural restriction on how mechanisms can react when agents change their opinion about the relative ranking of any two adjacently ranked groups of alternatives. 
Our result extends a result from \citep{MennleSeuken2017PSP_WP}, a decomposition of strategyproofness for strict preferences, to the full domain that includes weak preferences. 
\end{abstract}

\medskip
\noindent\textbf{JEL:} %\edit{[JEL]}%
%
%Vulnerability Paper:
%** C78: Bargaining Theory; Matching Theory,
C79, % Game Theory and Bargaining Theory - Other,
D82 %: Asymmetric and Private Information; Mechanism Design,
%H75: State and Local Government: Health; Education; Welfare; Public Pensions,
%I21: Analysis of Education,
%I28: Government Policy
%
%Carroll Local Sufficiency:
%D02: Institutions: Design, Formation, and Operations,
%D71: Social Choice; Clubs; Committees; Associations
%
%PS in large markets (Kojima, Manea):
%C70: Analysis of Collective Decision-Making - General;
%D61: Allocative Efficiency; Cost–Benefit Analysis;
%D63: Equity, Justice, Inequality, and Other Normative Criteria and Measurement
%
%Market Design D47
%
%
%C78, % Bargaining Theory, Matching Theory,
%I21 Analysis of Education;
%D61 Welfare econ;
%I20 Education;
%D82 Asymmetric and Private Information, Mechanism Design;
%D47, % Market Design,
%D82% Asymmetric and Private Information; Mechanism Design
%D78 Positive Analysis of Policy Formulation and Implementation
%H75 State and Local Government: Health, Education, Welfare, Public Pensions
%
}

\section{Introduction}
Ordinal mechanisms are commonly used in markets where monetary transfers are prohibited or restricted. 
Prominent examples include the assignment of seats at public schools, which is usually based on rank ordered preference lists from parents, and voting schemes, e.g., when the International Olympic Committee must agree on where to hold the next Olympic Games. 
A mechanism is said to be \emph{strategyproof} if it makes truthful reporting of their preference orders a dominant strategy for all agents. 
This is an important requirement for multiple reasons: 
first, strategyproof mechanisms are more likely to elicit truthful preferences from agents who act in their own best interest. 
This information can then be used to determine an appealing outcome (subject to the limitations imposed by strategyproofness); 
but it may also be useful beyond the role as input to the mechanism, e.g., to learn the true demand for particular schools in school choice settings. 
Second, it makes participation in the mechanism simple for the agents because they do not have to reason about the preferences or equilibrium strategies of the other agents. 
Therefore, strategyproofness is also a fairness requirement as it levels the playing field between agents with varying cognitive and computational capabilities. 

The standard definition of strategyproofness is composed of a series of incentive constraints: 
it requires that 
\emph{for any profile of true preferences, any agent, and any conceivable misreport by this agent, the agent should weakly prefer the outcome obtained from reporting their preferences truthfully to the outcome from submitting the misreport.}

First, observe that the number of constraints my be very large. 
If agents have strict preferences over $m$ alternatives, they may have any of $m!$ possible strict preference orders and submit any of possible $m!-1$ misreports. 
This means that the basic definition of strategyproofness involves an exponential number of  individual constraints for any agent in any conceivable situation. 
The number is even larger when they may also be indifferent between alternatives. 
This large number of constraints is problematic because it makes the strategyproofness concept unwieldy. 
Therefore, proving general statements about strategyproof mechanisms often requires non-trivial arguments; 
and the large number makes encoding strategyproofness as constraints to an optimization problem infeasible under the automated mechanism design paradigm.  

Second, observe that the restrictions imposed by strategyproofness are implicit and they yield few insights about the structure of strategyproof mechanisms. 
Thus, it is unclear how exactly these mechanisms look like, and it may be challenging to verify or disprove strategyproofness of a give mechanism or a given class of mechanisms. 

In this note, we address these challenges: 
we study the strategyproofness requirement for ordinal mechanisms on the full domain of weak preferences. 
Our main contribution is a decomposition of strategyproofness into three simple axioms.%
\footnote{This result is a direct generalization of the decomposition of strategyproofness for probabilistic assignment mechanisms on the strict preference domain into swap monotonicity, upper invariance, and lower invariance; 
see Theorem 1 in \citep{MennleSeuken2017PSP_WP}.} %
This decomposition mitigates the two concerns about strategyproofness described above. 
First, it provides a substantially smaller set of conditions that is equivalent to strategyproofness. 
Second, the axioms describe the ways in which strategyproof mechanisms may react to certain basic kinds of misreports and thereby delivers a much clearer picture of what strategyproof mechanisms look like. 
\section{Model}
Let $M$ be a set of $m$ \emph{alternatives}. 
An agent's \emph{preference order} is a weak order relation $R$ on the set of alternatives: 
$a ~ R ~ b$ means that the agent \emph{prefers} alternative $a$ to alternative $b$. 
The agent is \emph{indifferent} between $a$ and $b$ if $a ~R~ b$ and $b ~R~ a$ (denoted $a ~I~ b$), 
and the agent \emph{strictly prefers} $a$ to $b$ if $a ~R~ b$ but not $b ~R~ a$ (denoted $a~P~b$). 
We represent the preference order $R$ as 
\begin{equation}
	M_1 ~P~ \ldots ~P~ M_k ~P~ \ldots ~P~ M_K, 
\end{equation}
where $(M_k)_{1\leq k\leq K}$ is a partition of $M$ such that
\begin{itemize}
\setlength{\itemsep}{0pt}
	\item $a ~I~ b$ for all $a,b \in M_k$ and all $k\in \{1,\ldots,K\}$, 
	\item $a ~P~ b$ for any $a \in M_k, b\in M_{k+1}$ for any $k\in \{1,\ldots,K-1\}$. 
\end{itemize}
For the sake of notational simplicity, we formulate our results for situations with a single agent; 
however, they all extend straightforwardly to settings with multiple agents. 
Let $\mathcal{R}$ be the set of all possible preference orders, then a \emph{mechanism} is a mapping from a preference order to a lottery over alternatives; 
formally, $\varphi : \mathcal{R} \rightarrow \Delta(M)$. 

For any lottery $x \in \Delta(M)$ and any subset $A\subseteq M$ of the alternatives, let $x_A = \sum_{a \in A} x_a$ denote the probability of selecting an alternative within $A$. 
Given a preference order $R$ and two lotteries $x,y \in \Delta(M)$, we say that \emph{$x$ first order-stochastically dominates $y$ at $R$} if, for all alternatives $a \in M$, we have
\begin{equation}
	x_{\{j \in M: j~R~a\}} = \sum_{j \in M : j~R~a} x_j \geq \sum_{j \in M : j~R~a} y_j = y_{\{j \in M: j~R~a\}}. 
\end{equation}
A mechanism $\varphi$ is \emph{strategyproof} if, 
for all pairs of preference orders $(R,R') \in \mathcal{R}^2$, the lottery $\varphi(R)$ first order-stochastically dominates the lottery $\varphi(R')$ at $R$. 
\section{The Axioms}
We now define the axioms which make up our decomposition of strategyproofness. 
Each axiom restricts the way in which the mechanisms can react to particular changes in the preference report of the agent. 
A \emph{separation} is a pair of preference orders $(R,R') \in \mathcal{R}^2$ such that there exists some $\kappa \in \{1,\ldots,K\}$ with 
\begin{equation*}
\begin{array}{ccccccccccccc}
	M_1 & ~P~  & \ldots & ~P~  & M_{\kappa-1} & ~P~  & M_\kappa              & ~P~  & M_{\kappa+1} & ~P~  & \ldots & ~P~  & M_K, \\
	M_1 & ~P'~ & \ldots & ~P'~ & M_{\kappa-1} & ~P'~ & M^1_\kappa ~P'~ M^2_\kappa & ~P'~ & M_{\kappa+1} & ~P'~ & \ldots & ~P'~ & M_K, 
\end{array}
\end{equation*}
where $M^1_\kappa \dot{\cup} M^2_\kappa = M_\kappa$ is a disjoint partition of $M_\kappa$. 
Thus, the two preference orders of a separation $(R,R')$ differ only on the relative ranking of the alternatives in $M_\kappa$: 
an agent with preference order $R$ is completely indifferent between the alternatives in $M_\kappa$, whereas an agent with preference order $R'$ strictly prefers any alternative in $M^1_\kappa$ to any alternative in $M^2_\kappa$, but is indifferent between any two alternatives within $M^1_\kappa$ or within $M^2_\kappa$, respectively. 

Observe that a separation primarily reveals information about the agent's relative ranking of the alternatives in $M_\kappa$.  
The axioms we defined impose restrictions on how the outcomes of a mechanism can vary based on this information. 
\begin{axiom}
A mechanism $\varphi$ is \emph{separation responsive} if, 
for all separations $(R,R')$, we have
$\varphi_{M_\kappa^1}(P') \geq \varphi_{M_\kappa^1}(P)$ 
and 
$\varphi_{M_\kappa^2}(P') \leq \varphi_{M_\kappa^2}(P)$. 
\end{axiom}
Separation responsiveness captures the intuition that a mechanism should not assign \emph{less} probability to alternatives that the agent claims to \emph{prefer}. 
Under a separation responsive mechanism, these probabilities must change in the right direction if they change at all. 
\begin{axiom}
A mechanism $\varphi$ is \emph{separation direct} if, 
for all separations $(R,R')$ 
with $\varphi_{M_k}(R) \neq \varphi_{M_k}(R')$ for some $k \in \{1,\ldots,K\}$, 
we have  
$\varphi_{M_\kappa^1}(P') \neq \varphi_{M_\kappa^1}(P)$ 
and 
$\varphi_{M_\kappa^2}(P') \neq \varphi_{M_\kappa^2}(P)$. 
\end{axiom}
Intuitively, if a separation direct mechanism changes the outcome at all under some separation, then this change must affect at least the probability for the sets of alternatives for which differential preferences are reported, namely the alternatives in $M_\kappa^1$ relative to the alternatives in $M_\kappa^2$. 
\begin{axiom}
A mechanism $\varphi$ is \emph{separation monotonic} if it is separation responsive and separation direct. 
\end{axiom}
The intuition for separation monotonicity arises from the interpretations of separation responsiveness and directness: 
if a separation leads to any change in the outcome of the mechanism, then this change must affect at least the set of alternatives that was separated, and the change has to point in the right direction.\footnote{Separation monotonicity can be understood as swap monotonicity from \citep{MennleSeuken2017PSP_WP} but adjusted for the domain weak preferences; 
instead of the swaps of two adjacently ranked alternatives, the notion of \emph{locality} in this domain is described by separations.}%
\begin{axiom}
A mechanism $\varphi$ is \emph{separation upper invariant} if, 
for all separations $(R,R')$, we have 
$\varphi_{M_k}(P) = \varphi_{M_k}(P')$ for all $k \in \{1,\ldots,\kappa-1\}$.
\end{axiom}
Separation upper invariance yields robustness to a certain kind of strategic misreport. 
Suppose that the agent is interested primarily in the probability of higher ranking alternatives. 
If a mechanism is not separation upper invariant, then the agent may improve the changes of a higher ranking alternative by performing one or multiple separations of lower ranking sets of alternatives. 
In the domain of strict preferences, separation upper invariance corresponds to upper invariance, which is equivalent to truncation robustness \citep{Hashimoto2013TwoAxiomaticApproachestoPSMechanisms}. 
\begin{axiom}
A mechanism $\varphi$ is \emph{separation lower invariant} if, 
for all separations $(R,R')$, 
we have 
$\varphi_{M_k}(P) = \varphi_{M_k}(P')$ for all $k \in \{\kappa+1,\ldots,K\}$.
\end{axiom}
Separation lower invariance is the natural complement of separation upper invariance on the lower contour set. 
It requires that changes in the preference order over higher ranking alternatives do not influence the probabilities for sets of lower ranking alternatives. 
\section{Decomposition Results}
In this section, we present our main result, the decomposition of strategyproofness. 
\begin{theorem}
\label{THM:DECOMPOSITION}
A mechanism $\varphi$ is strategyproof if and only if it is 
separation monotonic, 
separation upper invariant, 
and separation lower invariant. 
\end{theorem}
\begin{remark}
In the domain with indifferences, the number of possible preference orders is $a(m)$, the $m^\text{th}$ Fubini number, which grows super-exponentially in $m$. 
The decomposition result allows us to reduce the number of constraints that we need to verify for strategyproofness: 
we only need to consider separations, instead of arbitrary misreports, and this number is bounded by $2^m$ in the worst case (when the agent is in fact indifferent between all alternatives), but usually much lower. 
\end{remark}
\begin{remark} 
\label{REM:RELAX_MON_TO_RESP}
The requirement of separation monotonicity in Theorem \ref{THM:DECOMPOSITION} can be relaxed to separation responsiveness because separation upper and lower invariance imply separation directness (see Lemma \ref{LEM:SP_IMPLIES_SD}). 
\end{remark}
Theorem \ref{THM:DECOMPOSITION} implies an equivalent condition for strateyproofness on the smaller class of deterministic mechanisms. 
\begin{corollary} 
\label{COR:DET_EQUIV}
A deterministic mechanism $\varphi$ is strategyproof if and only if it is separation monotonic. 
\end{corollary}
\begin{proof} 
It suffices to show that for deterministic mechanisms, separation monotonicity implies separation upper and lower invariance. 
Given any separation $(R,R')$ for which $\varphi(R) \neq \varphi(R')$, separation monotonicity requires that $\varphi_{M_{\kappa}^1}(R') > \varphi_{M_{\kappa}^1}(R)$ and $\varphi_{M_{\kappa}^2}(R') < \varphi_{M_{\kappa}^2}(R)$. 
Since $\varphi$ is deterministic, there must exist alternatives $a \in M_{\kappa}^1, b \in M_{\kappa}^2$ such that $\varphi_{\{a\}}(R') = \varphi_{\{b\}}(R) = 1$ and $\varphi_{\{a\}}(R) = \varphi_{\{b\}}(R') = 0$. 
In other words, $\varphi$ selects $b$ for preference order $R$ and $a$ for $R'$. 
Thus, there is no change in the probabilities of selecting any other alternatives. 
\end{proof}
%
%Observe that Corollary \ref{COR:DET_EQUIV} corresponds to Proposition 1 in \citep{MennleSeuken2017PSP_WP} but on the larger domain of preferences with indifferences. 
%
%
%\section{Conclusion}

\section*{Proof of Theorem \ref{THM:DECOMPOSITION}}
To prove Theorem \ref{THM:DECOMPOSITION}, we first show that strategyproofness of $\varphi$ implies the axioms (Lemmas \ref{LEM:SP_IMPLIES_SUI_SLI}, \ref{LEM:SP_IMPLIES_SD}, \ref{LEM:SP_IMPLIES_SR}). 
Subsequently, we show the more complicated sufficiency of the axioms for strategyproofness (Lemmas \ref{LEM:AXIOMS_IMPLY_SEP_SP}, \ref{LEM:AXIOMS_IMPLY_L_SEP_SP}, \ref{LEM:AXIOMS_IMPLY_MULTI_SEP_SP}, \ref{LEM:EXISTS_MULTI_SEP_TRANSITION}, \ref{LEM:MULTI_SEP_SP_IMPLIES_SP}). 
\begin{lemma} 
\label{LEM:SP_IMPLIES_SUI_SLI} 
If a mechanism $\varphi$ is strategyproof, then it is separation upper invariant and separation lower invariant. 
\end{lemma}
\begin{proof}[Proof of Lemma \ref{LEM:SP_IMPLIES_SUI_SLI}] 
We first show separation upper invariance. 
Towards contradiction, assume that $\varphi$ is strategyproof but not separation upper invariant. 
Then there exists a separation $(R,R')$, such that 
\begin{equation}
	\varphi_{M_{\tilde{k}}}(R) \neq \varphi_{M_{\tilde{k}}}(R'),
	\label{EQ:FIRST_VIOLATION_OF_SUI}
\end{equation}
where $\tilde{k} < \kappa$ is the smallest index for which inequality (\ref{EQ:FIRST_VIOLATION_OF_SUI}) holds. 
If $\varphi_{M_{\tilde{k}}}(R) < \varphi_{M_{\tilde{k}}}(R')$, then 
\begin{equation}
	\sum_{k=1}^{\tilde{k}} \varphi_{M_k}(R) < \sum_{k=1}^{\tilde{k}} \varphi_{M_k}(R'),
\end{equation}
which implies that $\varphi(R)$ does not first order-stochastically dominate $\varphi(R')$ at $R$, a contradiction to strategyproofness of $\varphi$. 
Conversely, if $\varphi_{M_{\tilde{k}}}(R) > \varphi_{M_{\tilde{k}}}(R')$, then an analogous argument implies that $\varphi(R')$ does not first order-stochastically dominate $\varphi(R)$ at $R'$, again a contradiction. 

The proof of separation lower invariance is analogous, except that we choose $\tilde{k} > \kappa$ to be the \emph{greatest} index for which $\varphi_{M_{\tilde{k}}}(R) \neq \varphi_{M_{\tilde{k}}}(R')$. 
\end{proof}
\begin{lemma} 
\label{LEM:SP_IMPLIES_SD}
If a mechanism $\varphi$ is strategyproof, then it is separation direct. 
\end{lemma}
\begin{proof}[Proof of Lemma \ref{LEM:SP_IMPLIES_SD}] 
By Lemma \ref{LEM:SP_IMPLIES_SUI_SLI}, $\varphi$ is separation upper and lower invariant. 
Thus, for any separation $(R,R')$ and any $k \neq \kappa$, we have $\varphi_{M_k}(R) = \varphi_{M_k}(R')$. 
Since 
\begin{equation}
	\varphi_{M_{\kappa}}(R)-\varphi_{M_{\kappa}}(R') = \sum_{k\neq \kappa} \varphi_{M_{k}}(R')-\varphi_{M_{k}}(R) = 0, 
\end{equation}
the conditions under which separation directness is any restriction, never arise under $\varphi$. 
Thus, this axiom is trivially satisfied. 
\end{proof}
\begin{lemma} 
\label{LEM:SP_IMPLIES_SR}
If a mechanism $\varphi$ is strategyproof, then it is separation responsive. 
\end{lemma}
\begin{proof} 
Towards contradiction, assume that $\varphi$ is strategyproof but not separation responsive. 
Lemma \ref{LEM:SP_IMPLIES_SUI_SLI} implies that, for any separation $(R,R')$, we have
\begin{equation}
	\varphi_{M_{k}}(R) = \varphi_{M_{k}}(R') 
\end{equation}
for all $k \in \{1,\ldots,K\}$ and
\begin{equation}
	\varphi_{M_{\kappa}^1}(R') - \varphi_{M_{\kappa}^1}(R) = \varphi_{M_{\kappa}^2}(R) - \varphi_{M_{\kappa}^2}(R'). 
\end{equation}
In particular, there exists a separation $(R,R')$, such that 
\begin{equation}
	\varphi_{M_{\kappa}^1}(R') - \varphi_{M_{\kappa}^1}(R) = \varphi_{M_{\kappa}^2}(R) - \varphi_{M_{\kappa}^2}(R') < 0. 
\end{equation}
Thus, 
\begin{equation}
	\varphi_{M_{\kappa}^1}(R') + \sum_{k=1}^{\kappa-1} \varphi_{M_{k}}(R') 
	< 
	\varphi_{M_{\kappa}^1}(R) + \sum_{k=1}^{\kappa-1} \varphi_{M_{k}}(R),
\end{equation}
which means that $\varphi(R')$ does not first order-stochastically dominate $\varphi(R)$ at $R'$, a contradiction to strategyproofness. 
\end{proof}
This establishes the sufficiency-part in Theorem \ref{THM:DECOMPOSITION}. 

\medskip
To show necessity, we introduce 
$L$-separations and multi-separations (Definition \ref{DEF:L_AND_MULTI_SEPARATION}), 
as well as separation-, $L$-separation-, and multi-separation strategyproofness (Definition \ref{DEF:SEP_L_SEP_MULTI_SEP_SP}). 
We then show that the axioms imply multi-separation strategyproofness (Lemmas \ref{LEM:AXIOMS_IMPLY_SEP_SP}, \ref{LEM:AXIOMS_IMPLY_L_SEP_SP}, \ref{LEM:AXIOMS_IMPLY_MULTI_SEP_SP}), 
which in turn implies strategyproofness (Lemmas \ref{LEM:EXISTS_MULTI_SEP_TRANSITION}, \ref{LEM:MULTI_SEP_SP_IMPLIES_SP}). 
\begin{definition}[$L$- and Multi-separation] 
\label{DEF:L_AND_MULTI_SEPARATION}
Let $(R,R')$ be a pair of preference orders:
\begin{itemize}
\setlength{\itemsep}{0pt}
	\item $(R,R')$ is called an \emph{$L$-separation} if there exist $L\geq 2$ and $\kappa \in \{1,\ldots,K\}$, such that
		\begin{equation*}
			\begin{array}{ccccccccc}
				M_1 & ~P~  & \ldots & ~P~  & M_\kappa              & ~P~  & \ldots & ~P~  & M_K, \\
				M_1 & ~P'~ & \ldots & ~P'~ & M^1_\kappa ~P'~ \ldots ~P'~ M^L_\kappa & ~P'~ & \ldots & ~P'~ & M_K, 
			\end{array}
		\end{equation*}
		where $M_{\kappa} = \bigcup_{l = 1}^L M^{l}_{\kappa}$ is a partition of $M_{\kappa}$ 
		(i.e., $M^{l}_{\kappa} \cap M^{l'}_{\kappa} = \emptyset$ for $l \neq l'$). 
	\item $(R,R')$ is called an \emph{$L$-separation} if there exist $L_1,\ldots,L_K\geq 1$, such that
		\begin{equation*}
			\begin{array}{ccccccccc}
				M_1 & ~P~  & \ldots & ~P~ & M_K, \\
				M_1^1 ~P'~ \ldots ~P'~ M_1^{L_1} & ~P'~ & \ldots & ~P'~ & M_K^1 ~P'~ \ldots ~P'~ M_K^{L_K}, 
			\end{array}
		\end{equation*}
		where, for each $k\in\{1,\ldots,K\}$, $M_k = \bigcup_{l = 1}^{L_k} M^{l}_{k}$ is a partition of $M_{k}$.  
\end{itemize}
\end{definition}
\begin{definition}[Separation Strategyproofness]
\label{DEF:SEP_L_SEP_MULTI_SEP_SP}
A mechanism $\varphi$ is \emph{separation/$L$-separation/multi-separation strategyproof} if for any separation/$L$-separation/multi-separation $(R,R')$ we have that 
\begin{enumerate}
\setlength{\itemsep}{0pt}
	\item $\varphi(R)$ first order-stochastically dominates $\varphi(R')$ at $R$, 
	\item $\varphi(R')$ first order-stochastically dominates $\varphi(R)$ at $R'$.
\end{enumerate}
\end{definition}
\begin{lemma} 
\label{LEM:AXIOMS_IMPLY_SEP_SP}
If a mechanism $\varphi$ is separation monotonic, separation upper invariant, and separation lower invariant, then it is separation strategyproof.
\end{lemma}
\begin{proof} 
Fix a separation $(R,R')$. 
By separation upper and lower invariance, we get that 
$\varphi_{M_k}(R) = \varphi_{M_k}(R')$ for all $k \in \{1,\ldots,K\}$. 
Thus, $\varphi(R)$ first order-stochastically dominates $\varphi(R')$ at $R$ trivially. 
In addition, by separation responsiveness, we have
\begin{equation}
\varphi_{M_{\kappa}^1}(P') - \varphi_{M_{\kappa}^1}(P) + \sum_{k=1}^{\kappa-1} \varphi_{M_k}(P') - \varphi_{M_k}(P) \geq 0. 
\end{equation}
This implies that $\varphi(R')$ first order-stochastically dominates $\varphi(R)$ at $R'$.
\end{proof}
\begin{lemma} 
\label{LEM:AXIOMS_IMPLY_L_SEP_SP}
If a mechanism $\varphi$ is separation monotonic, separation upper invariant, and separation lower invariant, then it is $L$-separation strategyproof for any $L\geq 2$.
\end{lemma}
\begin{proof} 
Consider an $L$-separation $(R,R')$. 
As in Lemma \ref{LEM:AXIOMS_IMPLY_SEP_SP}, first order-stochastic dominance of $\varphi(R)$ over $\varphi(R')$ at $R$ follows from separation upper and lower invariance. 

We now need to verify first order-stochastic dominance of $\varphi(R')$ over $\varphi(R)$ at $R'$ as well. 
For this, we need to show that 
\begin{equation}
	\sum_{l=1}^{\overline{l}} \varphi_{M_{\kappa}^l}(P') - \varphi_{M_{\kappa}^l}(P)  \geq 0
\label{EQ:FOSD_FOR_L_SEPARATION_VERFIY}
\end{equation}
for all $\overline{l} \in \{1,\ldots,L\}$. 
Observe that any $L$-separation can be decomposed into a sequence of $L-1$ separations $((R,R^1),(R^1,R^2),\ldots,(R^{L-2},R'))$, where 
\begin{equation*}
	\begin{array}{lllll}
		\ldots & ~P~  & \left( M^1_\kappa \cup \ldots \cup M^L_\kappa\right) & ~P~  & \ldots, \\
		\ldots & ~P^1~  & \left( M^1_\kappa \right) ~P^1~ \left(M^2_\kappa \cup \ldots \cup M^L_\kappa\right) & ~P^1~  & \ldots, \\		
		\ldots & ~P^2~  & \left( M^1_\kappa \right) ~P^2~ \left(M^2_\kappa \right) ~P^2~ \left(M^3_\kappa \cup \ldots \cup M^L_\kappa\right) & ~P^2~  & \ldots, \\
		\ldots & & & & \\
		\ldots & ~P^{L-2}~  & \left( M^1_\kappa \right) ~P^{L-2}~ \ldots ~P^{L-2}~ \left(M^{L-1}_\kappa \cup M^L_\kappa\right) & ~P^{L-2}~  & \ldots, \\
		\ldots & ~P'~ & \left(M^1_\kappa\right) ~P'~ \ldots ~P'~ \left(M^L_\kappa\right) & ~P'~ & \ldots. 
	\end{array}
\end{equation*}
By separation responsiveness, we have $\varphi_{M_\kappa^1}(R) \leq \varphi_{M_\kappa^1}(R^1)$, and separation upper invariance implies $\varphi_{M_\kappa^1}(R^1) = \varphi_{M_\kappa^1}(R^2) = \ldots = \varphi_{M_\kappa^1}(R')$. 
Thus, 
\begin{equation}
	\varphi_{M_\kappa^1}(R') - \varphi_{M_\kappa^1}(R) \geq 0. 
\end{equation}
Next, we consider a different sequence of separations, where 
\begin{equation*}
	\begin{array}{lllll}
		\ldots & ~P~  & \left( M^1_\kappa \cup \ldots \cup M^L_\kappa\right) & ~P~  & \ldots, \\
		\ldots & ~P^1~  & \left( M^1_\kappa \cup M^2_\kappa\right) ~P^1~ \left(M^3_\kappa \cup \ldots \cup M^L_\kappa\right) & ~P^1~  & \ldots, \\		
		\ldots & ~P^2~  & \left( M^1_\kappa \cup M^2_\kappa\right) ~P^2~ \left(M^3_\kappa \right) ~P^2~ \left(M^3_\kappa \cup \ldots \cup M^L_\kappa\right) & ~P^2~  & \ldots, \\
		\ldots & & & & \\
		\ldots & ~P^{L-3}~  & \left( M^1_\kappa \cup M^3_\kappa\right) ~P^{L-3}~ \ldots ~P^{L-3}~ \left(M^{L-1}_\kappa \cup M^L_\kappa\right) & ~P^{L-3}~  & \ldots, \\
		\ldots & ~P^{L-2}~  & \left( M^1_\kappa \cup M^2_\kappa\right) ~P^{L-2}~ \ldots ~P^{L-2}~ \left(M^L_\kappa\right) & ~P^{L-2}~  & \ldots, \\
		\ldots & ~P'~ & \left(M^1_\kappa\right) ~P'~ \left(M^2_\kappa\right) ~P'~ \ldots ~P'~ \left(M^L_\kappa\right) & ~P'~ & \ldots. 
	\end{array}
\end{equation*}
Again by separation responsiveness, we have $\varphi_{M_\kappa^1 \cup M_\kappa^2}(R) \leq \varphi_{M_\kappa^1 \cup M_\kappa^2}(R^1)$, separation upper and lower invariance imply $\varphi_{M_\kappa^1 \cup M_\kappa^2}(R^1)= \ldots = \varphi_{M_\kappa^1 \cup M_\kappa^2}(R')$. 
Thus, 
\begin{equation}
	\sum_{l=1}^2 \varphi_{M_\kappa^l}(R') - \varphi_{M_\kappa^l}(R) \geq 0. 
\end{equation} 
We can verify (\ref{EQ:FOSD_FOR_L_SEPARATION_VERFIY}) analogously for all $\overline{l} \in \{3,\ldots,L\}$. 
\end{proof}

\begin{lemma} 
\label{LEM:AXIOMS_IMPLY_MULTI_SEP_SP}
If a mechanism $\varphi$ is separation monotonic, separation upper invariant, and separation lower invariant, then it is multi-separation strategyproof.
\end{lemma}
\begin{proof} 
As in the proof of Lemma \ref{LEM:AXIOMS_IMPLY_L_SEP_SP}, first order-stochastic dominance of $\varphi(R)$ over $\varphi(R')$ at $R$ follows from separation upper and lower invariance. 
Observe that any multi-separation can be decomposed into a sequence of $L$-separations. 
Thus, arguments similar to those in Lemma \ref{LEM:AXIOMS_IMPLY_L_SEP_SP} yield first order-stochastic dominance of $\varphi(R')$ over $\varphi(R)$ at $R'$. 
%
%
%******
%As for $L$-separation strategyproofness, $f(t)$ ordinally dominates $f(t')$ at $t$ by separation upper and lower invariance. For ordinal dominance of $f(t')$ over $f(t)$ at $t'$, we need to show that 
%\begin{equation}
	%\left(\sum_{k=1}^{\bar{k}-1} \sum_{l = 1}^{L_{k}} \Delta(M_{k}^{(l)}) \right) + \left( \sum_{l = 1}^{\bar{l}} \Delta(M_{\bar{k}}^{(l)}) \right) \geq 0
%\end{equation}
%for all $\bar{k} \in \{1, \ldots, K\}$ and $\bar{l} \in \{1, \ldots, L_{\bar{k}}\}$. 
%For a given $\bar{k}$, consider the $L_{\bar{k}}$-separation $(t,t^{\bar{k}}$ with
%\begin{eqnarray}
	%t^{\bar{k}} & : & M_1 \succ \ldots \succ M_{\bar{k}-1} \succ M_{\bar{k}}^{(1)} \succ \ldots \succ M_{\bar{k}}^{(L_{\bar{k}})} \succ M_{\bar{k}+1} \succ \ldots \succ M_K,
%\end{eqnarray}
%which is reached via a sequence of separations that split the set $M_{\bar{k}}$, but leave the other sets unchanged. 
%From Lemma \ref{lem:L_sep_sp} we know that 
%\begin{equation}
	%\sum_{l = 1}^{\bar{l}} \Delta^f_{t,t^{\bar{k}}}(M_{\bar{k}}^{(l)}) \geq 0
	%\label{eq:local_step}
%\end{equation}
%for all $\bar{l} \in \{1,\ldots,L_{\bar{k}}\}$. But by separation upper and lower invariance, the allocation for any of the $M_{\bar{k}}^{(l)}$ does not change when other sets $M_{k},k \neq \bar{k}$ are separated. Thus, $\Delta^f_{t,t^{\bar{k}}}(M_{\bar{k}}^{(l)}) = \Delta^f_{t,t'}(M_{\bar{k}}^{(l)})$ and  
%\begin{equation}
	%\sum_{l = 1}^{\bar{l}} \Delta(M_{\bar{k}}^{(l)}) = \sum_{l = 1}^{\bar{l}} \Delta^f_{t,t'}(M_{\bar{k}}^{(l)}) \geq 0,
%\end{equation}
%i.e., we get ordinal dominance of $f(t')$ over $f(t)$ at $t'$. 
\end{proof}
\begin{definition}
\label{DEF:UTILITY_FUNCTION}
A mapping $u:M \rightarrow \mathds{R}^+$ is called a \emph{utility function}. 
$u$ is said to be \emph{consistent with} a preference order $R$ if $u(a) \geq u(b)$ whenever $a~R~b$, denoted $u \sim R$. 
\end{definition}
\begin{lemma} 
\label{LEM:EXISTS_MULTI_SEP_TRANSITION}
For any preference orders $R,R'$ and consistent utility functions $u \sim R$ and $u'\sim R'$, the line segment described by 
\begin{equation}
	\{ u_{\alpha} = (1-\alpha) u + \alpha u', \alpha \in [0,1]\}
\end{equation}
passes through a sequence of preference orders $R=R^0,R^1,\ldots,R^{S-1},R^S = R'$, such that for all $s\in \{0,\ldots,S-1\}$ either the pair $(R^s,R^{s+1})$ or the pair $(R^{s+1},R^{s})$ is a multi-separation. 
\end{lemma}
\begin{proof} 
For each $R^s$ through which the line $\{u_{\alpha}: \alpha \in [0,1]\}$ passes, let $\alpha_s \in [0,1]$ be such that $u_{\alpha_s}\sim R^s$. 
Assume towards contradiction that, for some $s$, neither the pair $(R^s, R^{s+1})$ nor the pair $(R^{s+1}, R^{s})$ is a multi-separation.  
Then there must exist objects $a,b,c,d$ (not necessarily different) such that 
$a ~P^s~ b$ and $c ~I^s~ d$ 
but 
$a ~I^{s+1}~ b$ and $c ~P^{s+1}~ d$. 
This means that
\begin{center}
$\begin{array}{cc}
	u_{\alpha_s} (a) > 	u_{\alpha_s} (b), & u_{\alpha_s} (c) = u_{\alpha_s} (d), \\
	u_{\alpha_{s+1}} (a) = 	u_{\alpha_{s+1}} (b), & u_{\alpha_{s+1}} (c) > u_{\alpha_{s+1}} (d). 
\end{array}$
\end{center}
Taking the ``average'' of $u_{\alpha_s}$ and $u_{\alpha_{s+1}}$, we get a new utility function
\begin{equation}
	\tilde{u} = u_{\frac{1}{2}\left(\alpha_s + \alpha_{s+1}\right)} = \frac{1}{2}\left( u_{\alpha_s} + u_{\alpha_{s+1}} \right),
\end{equation}
which lies on the line between $u$ and $u'$, but where 
\begin{eqnarray}
	& \tilde{u} (a) > \tilde{u} (b) \text{ and } \tilde{u} (c) > \tilde{u} (d). & 
\end{eqnarray}
Thus, the line passes through a different type $\tilde{R}$ with $\tilde{u} \sim \tilde{R}$ between $R^s$ and $R^{s+1}$, which contradicts the assumption that the line passes directly from $R^s$ to $R^{s+1}$.
\end{proof}
\begin{lemma} 
\label{LEM:MULTI_SEP_SP_IMPLIES_SP}
If a mechanism $\varphi$ is multi-separation strategyproof, then it is strategyproof. 
\end{lemma}
\begin{proof} 
The arguments in this proof are similar to the proof of the local sufficiency result in \citep{Carroll2012WhenAreLocalIncentiveConstraintsSufficient}. 

We use the fact that for any two lotteries $x,y \in \Delta(M)$ and any preference order $R$, $x$ weakly first order-stochastically dominates $y$ at $R$ if and only if, for any utility function $u \sim R$, we have 
\begin{equation}
	\left\langle u , x-y\right\rangle = \sum_{j\in M} u(j) \cdot \left(x_j - y_j \right) \geq 0.  
\end{equation}
For two preference orders $R,R'$ and consistent utility functions $u\sim R$ and $u' \sim R'$, 
let $\{u_\alpha : \alpha \in [0,1]\}$ be the line segment in the space of utility functions that connects $u$ and $u'$. 
Following Lemma \ref{LEM:EXISTS_MULTI_SEP_TRANSITION}, 
let $\alpha_s \in [0,1]$ be such that $u_{\alpha_s} \sim R^s$ (i.e., $u_{\alpha_s}$ is a utility function on the line segment and consistent with $R^s$), where $R^s$ is the respective element of the sequence of preference orders through which the line segment passes. 
By construction,
\begin{equation}
\alpha_{s+1} u_{\alpha_{s}} - \alpha_{s} u_{\alpha_{s+1}} = (\alpha_{s+1} - \alpha_s) \cdot u,
\end{equation}
and from multi-separation strategyproofness of $\varphi$, we get
\begin{equation}
\left\langle u_{\alpha_s} , \varphi(R^s) - \varphi(R^{s+1}) \right\rangle \geq 0\text{ and }\left\langle u_{\alpha_{s+1}} , \varphi(R^{s+1}) - \varphi(R^{s}) \right\rangle \geq 0.
\end{equation}
Therefore, for all $s \in \{0,\ldots,S-1\}$,
\begin{equation}
\left\langle u , \varphi(R^s) - \varphi(R^{s+1}) \right\rangle \geq 0.
\end{equation}
Summing over all $s$ yields
\begin{equation}
\left\langle u , \varphi(R) - \varphi(R') \right\rangle 
	= \left\langle u , \varphi(R^0) - \varphi(R^S) \right\rangle \geq 0,
\end{equation}
which concludes the proof. 
\end{proof}
Lemma \ref{LEM:MULTI_SEP_SP_IMPLIES_SP} concludes the proof of sufficiency of the axioms for strategyproofness. 
%
%

% Bibstyle aea.bst version 2009.05.20

\end{document}